\documentclass[acmsmall,nonacm]{acmart}
\pdfoutput=1

\setcopyright{acmcopyright}
\copyrightyear{2021}
\acmYear{2021}
\acmDOI{10.1145/1122445.1122456}

\usepackage[english]{babel}
\usepackage{blindtext}

\usepackage{caption}
\usepackage{subcaption}

\usepackage{bruce-common}

\newcommand*{\etal}{et al.\@\xspace}

\newcommand\BDP{\textsc{BDP}}
\newcommand\RTT{\textsc{RTT}}

\newcommand\fair{$\Delta$-almost fair}
\newcommand\fairness{$\Delta$-almost fairness}
\newcommand\fairSymbol{\Delta}
\newcommand\fairLower{\Delta^{-1}}
\newcommand\fairUpper{\Delta}

\graphicspath{ {figures/} }

\begin{document}
\title{Updating the Theory of Buffer Sizing}

\author{Bruce Spang}
\email{bspang@stanford.edu}

\author{Serhat Arslan}
\email{sarslan@stanford.edu}

\author{Nick McKeown}
\email{nickm@stanford.edu}

\renewcommand{\shortauthors}{Spang, Arslan, McKeown}

\begin{abstract}
    Routers have packet buffers to reduce packet drops during times of congestion. It is important to correctly size the buffer: make it too small, and packets are dropped unnecessarily and the link may be underutilized; make it too big, and packets may wait for a long time, and the router itself may be more expensive to build. Despite its importance, there are few guidelines for picking the buffer size. The two most well-known rules only apply to long-lived TCP Reno flows; either for a network carrying a single TCP Reno flow (the buffer size should equal the bandwidth-delay product, or $\BDP$) or for a network carrying $n$ TCP Reno flows (the buffer size should equal $\BDP/\sqrt{n}$). Since these rules were introduced, TCP Reno has been replaced by newer algorithms as the default congestion control algorithm in all major operating systems, yet little has been written about how the rules need to change. This paper revisits both rules. For the single flow case, we generalize the $\BDP$ rule to account for changes to TCP, such as Proportional Rate Reduction (PRR), and the introduction of new algorithms including Cubic and BBR. We find that buffers can be made 60-75\% smaller for newer algorithms.  For the multiple flow case, we show that the square root of $n$ rule holds under a broader set of assumptions than previously known, including for these new congestion control algorithms. We also demonstrate situations where the square root of $n$ rule does {\em not} hold, including for unfair flows and certain settings with ECN. We validate our results by precisely measuring the time series of buffer occupancy in a real network, and comparing it to the per-packet window size.
\end{abstract}

\maketitle

\section{Introduction}

Internet routers have packet buffers which reduce packet loss during times of congestion. Sizing the router buffer correctly is important: if a router buffer is too small, it can cause high packet loss and link under-utilization. If a buffer is too large, packets may have to wait an unnecessarily long time in the buffer during congested periods, often up to hundreds of milliseconds.
While an operator can reduce the operational size of a router buffer, the maximum size of a router buffer is decided by the router manufacturer, and the operator typically configures the router to use all the available buffers. Without clear guidance about how big a buffer needs to be, manufacturers tend to oversize buffers and operators tend to configure larger buffers than necessary, leading to increased cost and delay.

This paper revisits two widely used rules of thumb for sizing router buffers in the internet. The two rules cover two different cases: 

\noindent {\bf Case 1: When a network carries a single TCP Reno flow}. Van Jacobson observed in 1990 \cite{Jac1990} that a bottleneck link carrying a {\em single} TCP Reno flow requires a router buffer of size $B \ge \BDP$, the bandwidth-delay product, in order to keep the link fully utilized.

\noindent {\bf Case 2: When a network carries multiple TCP Reno flows}. Appenzeller, Keslassy, and McKeown argued in 2004 \cite{Appenzeller:2004fk} that a bottleneck link carrying $n$ long-lived TCP Reno flows requires a buffer of size $B \ge \BDP/\sqrt{n}$ in order to keep the link highly utilized.

Much has changed since these rules were first introduced, and it is not clear whether these rules still apply in modern networks. The behavior of TCP Reno has changed; most notably when Rate-Halving \cite{Hoe1995,SMM} and PRR~\cite{DMCG2011} were introduced. New types of congestion control have become widespread, such as Cubic~\cite{HRX2008a} (default in Linux, Android, and MacOS), and more recently BBR (deployed by Google for YouTube)~\cite{CCG+2017} and BBRv2 \cite{CCY+2019b}. Given that the analysis underlying both buffer sizing rules depends on the specific way in which TCP Reno halves the congestion window when losses are detected, there is no particular reason for either rule to still hold in today's internet.

Existing rules of thumb help us pick the buffer size to achieve full link utilization, and do not predict behavior if the buffer is made smaller. Thus, theory falls short for recent congestion control algorithms (e.g. BBR and BBRv2) which no longer aim to keep a bottleneck link running at 100\% utilization. Instead, they rely on short periods of under-utilization to keep queueing delay low and to estimate propagation delay.

In light of these changes, this paper examines buffer sizing for modern TCP algorithms. We show that the two rules still allow TCP Reno to fully utilize a link, despite changes due to Rate-Halving and PRR. We show that TCP Cubic, Scalable TCP, and BBR allows us to reduce buffer sizes. 

We extend our analysis for the case when link utilization is less than 100\%, and we show that very small buffers can still allow high (but not 100\%) link utilization. More generally, {\em this paper sheds new light on how to size buffers for a given congestion control algorithm and desired link utilization, under a very broad set of conditions}. In doing so, we also show how future congestion control algorithms can be designed to further reduce buffer requirements.

Throughout the paper, we will illustrate and validate our results using measurements drawn from a physical network in our lab. This is challenging: while Linux can capture per-packet measurements in the end-host TCP stack, it is not normally possible to capture the full time series of buffer occupancy at the switch. Our measurement setup uses a P4-programmable Tofino switch which we program to report the precise time evolution of its buffer, to approximately 1 nanosecond resolution. This lets us precisely compare the evolution of the congestion window and the buffer size and validate our theoretical results. We start in Section~\ref{sec:methodology} by describing this experimental setup.

In Section~\ref{sec:single-flow}, we revisit case \#1: a single TCP flow. Specifically, we show theoretically and experimentally that despite the introduction of PRR to TCP Reno, the $B \ge \BDP$ result still holds. We also show that with TCP Cubic  we can reduce the buffer size to $0.4\BDP$ for a single flow, and for BBR we can reduce it further to $0.25\BDP$ and still achieve full link utilization. We also describe how link utilization behaves when buffers are below these values.

Next, in Section~\ref{sec:multiple-flows}, we examine case \#2: multiple TCP flows.
We prove that algorithms which respond to full queues (via losses or marks) create standing queues with particular properties during times of congestion. We show that the square root of $n$ rule is a consequence of this behavior: if queues are always close to full for any buffer size, then intuitively the buffers can be shrunk without impacting utilization. Another, perhaps surprising, consequence is that link utilization is not a ``cliff'' function. We show that if the buffer is slightly too small for full utilization, then utilization still remains high. Specifically, we show that even with very small buffers, the utilization percent is at least $\Omega(1-1/\sqrt{n})$. We also show that design choices in BBR allow for a square root of $n$ rule essentially without assumptions.

In Section~\ref{sec:eval-sqrt-n} we support our results with fine-grained measurements from lab experiments. We show that the square root of $n$ rule holds in our experiment, and show how loss-based algorithms keep queues full during congestion. We verify the square root of $n$ result for BBR, and find experimentally that a $\BDP/n$ rule may be more accurate, allowing for even smaller buffers. We show that buffer sizes depend on how flows synchronize and whether they are fair. We describe how to check the assumptions in our theorems, and show that they hold in the lab. We also caution that enabling ECN can increase synchronization in certain cases, resulting in a larger required buffer, and show lab results to support our analysis.

In Section~\ref{sec:in-the-wild}, we discuss measurements of large-scale production networks which show that RTTs are persistently elevated during times of congestion~\cite{LDC+2014,FVD2017,DCG+2018}. Our results explain why congestion control algorithms create these persistently full queues, and suggest that these queues can be made smaller. %
We also discuss the implications of our results on other situations when it may be more difficult to measure congestion.

In Section~\ref{sec:recommendations} we make recommendations to network operators who are interested in running buffer sizing experiments, and describe how our results can be used to easily design these experiments. We also give recommendations to designers of new congestion control algorithms who would like to ensure small buffer requirements. We discuss related work in Section~\ref{sec:related} and conclude in Section~\ref{sec:conclusion}.

\noindent{\bf Contributions:} The main contributions of this paper are:
\begin{enumerate}
    \item Single flow case: A simple proof of TCP's required buffer size, applicable to the latest versions of TCP Reno, as well as other algorithms including Cubic, Scalable TCP, and BBR.
    \item Multiple flow case: A new, more general model of how buffer size is impacted by fairness and the amount of worst-case packet drops, and square root of $n$-style rules for TCP Reno and other algorithms.
    \item A better understanding of how congestion control algorithms interact with buffers, including how utilization depends on buffer size, how algorithms can reduce buffer requirements, and how current congestion measurement techniques rely on certain algorithmic behavior.
    \item A new measurement platform allowing precise observation of TCP and the router buffer.
\end{enumerate}

\begin{table*}
\begin{tabular}{|l|p{0.5\textwidth}|l|}
\hline
\textbf{Min. Buffer Size} & \textbf{Additional assumptions beyond Section~\ref{sec:prelim}} & \textbf{Citation}                                             \\ \hline
$\BDP$ & Reno, silence after loss & \cite{Jac1990,VS1994} \\
$\BDP$ & Reno & \cite{DHD2005}, Section~\ref{sec:new-bdp-rules} \\
$(1/b-1)\BDP$ & Multiplicative decrease by b, silence after loss & \cite{HR2008,LF2015,MAK2019a} \\
\rowcolor{lightgray} $(1/b-1)\BDP$ & Multiplicative decrease by b & Section~\ref{sec:new-bdp-rules}         \\
$\frac{3}{7} \BDP$ & Cubic & \cite{LF2015}, Section~\ref{sec:new-bdp-rules}         \\ 
\rowcolor{lightgray} $\tfrac{1}{7} \BDP$ & Scalable TCP & Section~\ref{sec:new-bdp-rules}         \\
\rowcolor{lightgray} $\tfrac{1}{4} \BDP$ & BBR during the probe bandwidth phase, with loss & Section~\ref{sec:new-bdp-rules}        \\
\hline
$\Theta(BDP/\sqrt{n})$ & Reno, windows are i.i.d. uniform random variables & \cite{Appenzeller:2004fk}                    \\
\rowcolor{lightgray} $O(\BDP/\sqrt{n})$ & $\sqrt{n} + O(n^2/BDP)$ almost fair flows see loss & Theorem~\ref{thm:our-sqrt-n-for-buffers}            \\
\rowcolor{lightgray} $O(\BDP/\sqrt{n})$ & Almost fair BBR flows in probe bandwidth phase & Theorem~\ref{thm:sqrt-n-bbr} \\ \hline
$O(p\cdot \BDP - n+np)$ & A $p$ fraction of fair flows see losses & \cite{DHD2005}    \\  
\rowcolor{lightgray} $O(s \cdot \BDP/n)$ & At most $s + n^2/BDP$ almost fair flows see loss. & Section~\ref{sec:desync}    \\  
\hline
$O(1)$ & Reno, bounded window size, Poisson pacing  & \cite{Enachescu:2006kp} \\ 
\rowcolor{lightgray} $O(1)$ & $\sqrt{n} + O(n^2/BDP)$ almost fair flows see loss, utilization is $\Omega(1-1/\sqrt{n})$ & Theorem~\ref{thm:utilization}  \\            
\hline
\end{tabular}
\caption{Minimum buffer sizes required for full link utilization, our new results highlighted in gray.}
\end{table*}

\section{Experiment Methodology}
\label{sec:methodology}

Throughout the paper, we will use measurements from our physical testbed to illustrate and validate our results. We have built a platform on programmable hardware which allows us to easily run experiments where TCP flows share a congested link, and observe how TCP and queues behave at a packet-by-packet level.

While working on this paper, we were frequently baffled by experimental results that didn't match our understanding of TCP. Almost always, when we dug into the measurements, we found that TCP's actual behavior did not match our understanding. We hope that by including our measurements, we can make TCP's behavior (and our results) more understandable.

\noindent\textbf{Setup} Our test network consists of two servers with 32 2.4Ghz cores and 32 GB of RAM each, connected by a Barefoot Tofino switch and use up to 2 MB of buffers. The servers run Linux 5.5.0, each with an Intel 82599ES 10Gb/s NIC. Each NIC is connected to a port of a 6.5Tb/s Barefoot Tofino switch via 100G to $4\times10$Gb/s breakout cables. The sender server is connected to the Tofino with two 10G cables. The interfaces are bonded and packets are equally split between them, which ensures that congestion happens at the switch (otherwise we only see congestion at the sender NIC). We set MTUs to 9000 bytes so the servers can sustain a 10Gb/s rate. We add 1ms of delay at the sender using Linux's traffic controller \texttt{tc}, and use iperf3 to generate TCP traffic. We used congestion control algorithms available in Linux 5.5.0, including TCP Reno, Cubic, BBR (v1), and Scalable TCP. We also used Google's alpha release of BBR2 \cite{CCY+2019b}.

\noindent\textbf{Measuring TCP} Since 2017, Linux has had a tracing system which reports the congestion window, smoothed RTT, and other important information for each received acknowledgement~\cite{Hir2017}. We used this system to observe TCP's behavior in our experiments, and have created a small open source tool to make them easier to work with: \url{https://github.com/brucespang/tcp_probe}.

\noindent\textbf{Measuring Buffers} It is usually hard to precisely measure a switch's buffer occupancy over time. We wrote a P4 program running on the Barefoot Tofino switch, to measure (and report) the time series of the packet buffer occupancy. Each time the switch receives a packet it sends a small UDP packet (often called a ``postcard'') to a collector which includes the current buffer occupancy and some packet metadata. At the collector, we record these postcards and match them with the TCP trace samples. Our ability to measure the complete time-series of a buffer is not novel (e.g. \cite{179783,BGG+2008,Kim:tnpU6lBL,CK2019,BKT+2019}). We believe this sort of fine-grained measurement should be a standard part of any TCP experiment.

\noindent\textbf{Datasets and Source Code} We have released our P4 code, infrastructure for collecting TCP traces, and the traces we collected online at \url{https://github.com/brucespang/ifip21-buffer-sizing }.

\section{Sizing buffers for a single flow}
\label{sec:single-flow}

Since the 1990s, it has been known that a network carrying a single TCP Reno flow requires a buffer size of one BDP in order to keep the bottleneck link fully utilized. However, TCP has changed dramatically since then, with Proportional Rate Reduction \cite{DMCG2011} and new congestion control algorithms like Cubic \cite{HRX2008a} and BBR \cite{CCG+2017}. It is not clear whether the buffer requirements for a single TCP flow have also changed.

In this section, we revisit the classic rule of thumb that a buffer should be at least one BDP for a single TCP connection. Some of the results in this section have been shown by prior authors, and we revisit their results to verify that they still hold experimentally, and as a warm-up for later sections. We find that despite the changes to the Linux kernel, the BDP rule still holds for TCP Reno. We find that a BDP is an \emph{overestimate} of the buffering required for Cubic and BBR.

More specifically, we will
\begin{enumerate}
\item Show that the BDP rule still holds for a modern implementation of TCP Reno.
\item Show that Cubic requires a buffer of only $\frac{3}{7} \BDP$ and BBR requires a buffer of only $\frac{1}{4} \BDP$.
\item Show how link utilization changes when buffers are smaller than required for 100\% utilization.
\item Show that BBR2 does not fully utilize a link, but instead aims for high link utilization.
\item Show how the usual proof of the BDP rule depends on outdated TCP behavior, and give a simpler and more general proof.
\end{enumerate}

\subsection{The original BDP rule}
The original BDP rule, attributed to Jacobson~\cite{Jac1990} and Villamizar and Song~\cite{VS1994}, states that for a single TCP Reno flow to fully utilize a link with a drop-tail buffer, the buffer must be at least the capacity of the link $C$ times the RTT. This is the bandwidth delay product, $\BDP = C \cdot RTT$.
\begin{fact}[BDP Rule]
  \label{thm:bdp-rule}
  TCP Reno fully utilizes a link if and only if $B \geq \BDP$.
\end{fact}

Most proofs of the BDP rule \cite{Jac1990,Appenzeller:2004fk} rely on TCP halving its window size $W$ on a loss, and then waiting to send another packet until it receives $W/2$ acknowledgements so that the number of packets in-flight drops below the new window size. But this is not how modern TCP implementations work. For example, TCP Cubic decreases its window by less than a half on a loss. Even TCP Reno no longer stops sending packets on a loss. Instead it behaves according to Proportional Rate Reduction (PRR)~\cite{DMCG2011}, introduced by Dukkipati \etal in 2011, which supplanted Rate-Halving, proposed by Hoe \cite{Hoe1995} in 1995 and as an RFC by Mathis \etal \cite{SMM} in 1999. Both algorithms gradually lower the size of the window after receiving a loss. When PRR decides to decrease the congestion window from $w$ to $w'$, it sends a new packet for every $w/w'$ packets acknowledged. For TCP Reno, this means a new packet is sent for every two packets acknowledged. Most deployments of TCP Reno now use PRR.

\subsection{Does the BDP rule still hold experimentally?}

Our measurement infrastructure allows us to observe the behavior of PRR after a loss, which is shown in Figure~\ref{fig:prr-loss} for TCP Reno (with PRR) and buffer size $B=\BDP$. The figure is zoomed in to show what happens when a packet is dropped by the switch at about the 1ms marker. The TCP source detects the loss about 1.5ms later, and begins to gradually decrease its window over the next 2ms. After 2ms (about 2 RTTs), the window size has been halved and the queue is nearly empty.

Without PRR, TCP Reno stops sending packets as soon as the drop is detected and remains silent until the window is reduced by half. Even though TCP with PRR never stops sending, the queue almost completely drains, which suggests that a BDP worth of buffering is still necessary. In what follows, we prove this.

\begin{figure}
  \centering
  \includegraphics[width=0.49\linewidth]{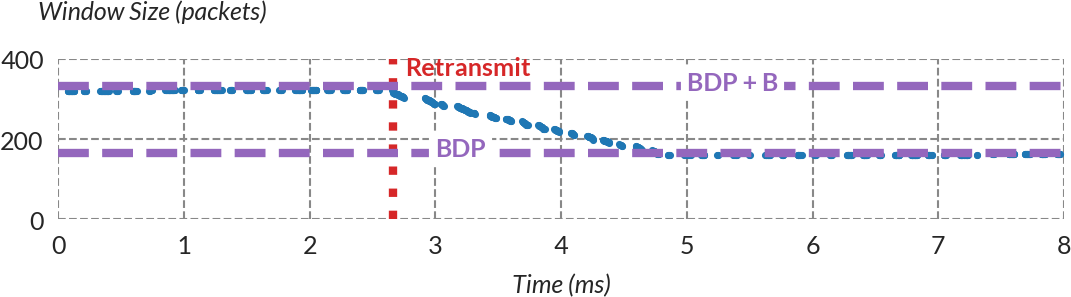}
  \hfill
  \includegraphics[width=0.49\linewidth]{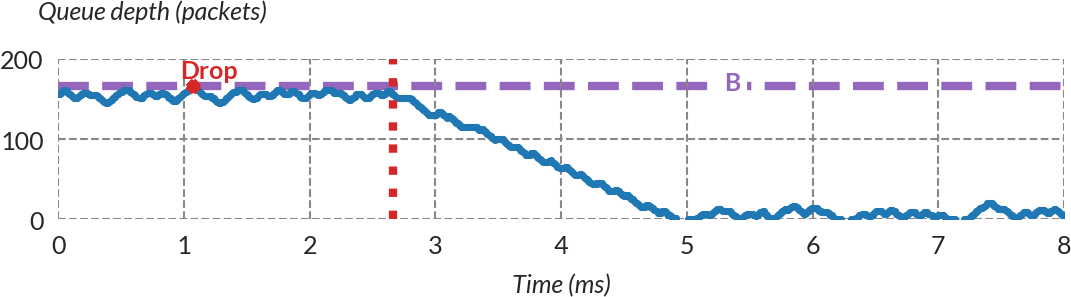}
  \caption{Per-packet window size (from kernel) and queue depth (from switch), for one TCP Reno flow with a BDP sized buffer. PRR keeps sending packets after a loss, but the queue still drains.}
  \label{fig:prr-loss}
\end{figure}

\subsection{Technical Preliminaries}
\label{sec:prelim}

In order to prove that $B\ge \BDP$ is still necessary with PRR, we need an argument that does not rely on TCP stopping sending after a loss. To this end, we set up some definitions that will help us prove our results throughout the paper. These are standard assumptions in the TCP and the buffer sizing literature, and apply to the common case when $n$ TCP flows share a drop-tail queue.

Throughout the paper, we will refer to a \emph{TCP flow} in theorem statements, and use it to mean the following set of assumptions and definitions, plus commonly known TCP behavior.

\begin{definition}
A packet is {\em in-flight} if it has been sent, but the acknowledgement has not yet been received by the sender. Let $W(t)$ be the aggregate window, or the number of in-flight packets at time $t$.

For the case of $n$ flows sharing a link, let $w_i(t)$ be the window of flow $i$, i.e. the number of in-flight packets for flow $i$ at time $t$.
Therefore the aggregate window is $W(t) = \sum_{i=1}^n w_i(t)$.
\end{definition}

Throughout the paper, we will heavily use the following behavior of TCP Reno: when TCP Reno successfully sends a packet and receives an acknowledgement, it increases its window from $w_i$ to $w_i+1$. When it detects a congestion signal (e.g. a dropped packet or an ECN mark), it decreases the window from $w_i$ to $w_i/2$.

Note that our definition of the window of flow $i$ is the number of in-flight packets, which is not necessarily the same as the congestion window ({\em cwnd}). For example, when TCP Reno reduces its congestion window in response to a loss, the number of in-flight packets may briefly exceed the congestion window.

\begin{definition}
The RTT at time $t$ is the time taken from when a packet is sent at time $t$ until we receive its acknowledgement. 
Our results assume that all flows have the same RTT.
\end{definition}

\begin{definition}
The \emph{sending rate} of flow $i$ at time $t$ is $x_i(t) = w_i(t)/RTT(t)$. The aggregate sending rate is $x(t) = W(t)/RTT(t)$.
\end{definition}

Links have queues to store packets before they are sent. We will also call them buffers. Packets are put into these queues as they arrive, and are later sent from the link. Queues have a length at any point in time, which we call the queue length, and a maximum number of packets they can fit before they need to drop packets, which we call the maximum queue length or buffer size.

\begin{definition}
  Let $Q(t)$ be the length of the queue at time $t$.
\end{definition}

A link with capacity $C$ sends one packet from the queue every $C^{-1}$ seconds, provided the queue is not empty. While some links, such as wireless links, have variable capacity, we only consider fixed data-rate links that send a packet every $C^{-1}$ seconds.

We will care about the utilization of these links.
\begin{definition}
For a link with departure rate $D(t)$ and capacity $C$, the link utilization is $\mu(t) = D(t)/C$
\end{definition}

A work-conserving queue will be fully utilized at time $t$ if and only if $Q(t) > 0$. While it may seem counter-intuitive, dropping packets does not necessarily reduce link utilization. If packets are dropped from the tail of the queue (i.e. an arriving packet doesn't fit, so it is dropped), then packets are lost \emph{before} the bottleneck link. If a work-conserving link always has a non-empty queue, it is fully utilized no matter how many packets are dropped.

\begin{definition}
The bandwidth-delay product of a flow crossing a link of capacity $C$ with round-trip time $RTT$ is $\BDP = C \cdot RTT.$
\end{definition}

We will use the following standard assumption about how the size of the queue relates to the number of in-flight packets~\cite{Appenzeller:2004fk,DHD2005,LF2015,AGM+2010,WG2006}.

\begin{assumption}
  \label{assm:queue}
  Let $B$ be the maximum buffer size, $t$ be some time, and $W(t)$ be the number of in-flight packets at time $t$. We assume that at time $t$, the queue size is
 $$Q(t) = \begin{cases}
    0 & \mbox{if}\quad W(t) \leq \BDP \\
    W(t) - \BDP & \mbox{if}\quad 0 < W(t) < B + \BDP\\
    B & \mbox{if}\quad W(t) \geq B + \BDP \\
  \end{cases}
  $$
\end{assumption}

Our results use this assumption heavily, and so it merits some discussion. It is a strong assumption, designed to let us compare our results to existing work. Consider the queue at some time $t$ when there are $W(t)$ packets in-flight. Over the previous RTT, if the bottleneck link is fully-utilized it will (by definition of BDP) send exactly one $\BDP$ worth of packets spaced  $C^{-1}$ seconds apart. Our assumption is essentially that the congestion control algorithm will closely match this behavior; i.e. it will also send packets spaced $C^{-1}$ seconds apart. In the case of TCP, this is the ACK clocking mechanism; packets sent from the bottleneck queue lead to a stream of ACKs at the sender, spaced $C^{-1}$ seconds apart, which pace out the sent packets.

Assumption~\ref{assm:queue} holds for the results in this paper because of ACK clocking, and is validated by our experimental results. If we were to model new congestion control algorithms, we would need to first check that this assumption holds. If, to pick an extreme example, a new congestion control algorithm required the sender to remain silent for a long time and then suddenly send a burst of packets, then the assumption would not hold. Fortunately, this is not how currently deployed congestion control algorithms work. 

We could relax this assumption if needed, provided we kept some relationship between queue length $Q(t)$ and $\BDP$. As an example, we observed some amount of variability in the queue depth in our measurements (e.g. Figure~\ref{fig:prr-loss}), which we believe is due to packet bursts from the sender. We could model this variability by assuming that $(1-\eps)(W(t) - \BDP) \leq Q(t) \leq (1+\eps)(W(t) - \BDP)$ for some constant $\eps > 0$. We could then derive similar results. For instance, the $B \geq \BDP$ rule for TCP Reno would become $B \geq (1+\eps)\BDP$. We have not carried this parameter through our results since it makes our theorems harder to understand and compare with existing work, and $\eps=0$ captures the general behavior we observe in our measurements.

With this assumption, we can prove the following well-known lemma which is key to showing all the buffer sizing results in this paper.
\begin{lemma}
  \label{lm:main}
  For TCP flows sharing a link at time $t$,
  the utilization $\mu(t) = \min(W(t)/\BDP, 1)$.
\end{lemma}
\begin{proof}  
\noindent If $W(t) > \BDP$, then applying Assumption~\ref{assm:queue} we get $Q(t) = \min(W(t) - \BDP, B) > 0$. We assume queues are work conserving, so $\mu(t) = 1$. If $W(t) \leq \BDP$ then during RTT $t$ we send $W(t)$ packets at a rate of $x(t) = W(t)/\RTT$. By Assumption~\ref{assm:queue}, $Q(t) = 0$, so the departure rate $D(t) = x(t)$. Therefore the utilization is $\mu(t) = D(t)/C = W(t)/\BDP$.
\end{proof}

Lemma~\ref{lm:main} is simultaneously obvious and a bit surprising. In words, if the number of in-flight packets {\em ever} falls below a BDP, then we will lose utilization. One might imagine that if the queue is full and the arrival rate is slightly less than $C$ (which happens if $W(t)$ is slightly less than a $\BDP$), the queue would only drain a little bit. But by Assumption~\ref{assm:queue} (and perhaps because of the way TCP's in-flight mechanism works), this means the sender will only send at a rate of at most $W(t)/RTT$ packets, which is slightly less than $C$.

Finally, we will need the following assumption about when losses happen.
\begin{assumption}[Loss]
  \label{assm:loss}
  Let $B$ be the maximum buffer size. We assume that TCP loses packets at time $t$ if and only if $Q(t) \geq B$, which by Assumption~\ref{assm:queue} is equivalent to $W(t) \geq \BDP + B$.
\end{assumption}

Lemma~\ref{lm:main} is easiest to use when a congestion control algorithm decreases its window in response to a full queue (e.g. a lost or marked packet), since Assumption~\ref{assm:loss} gives us a direct relationship between the buffer size and the number of in-flight packets upon a loss (or mark). In our experiments, Reno, Cubic, BBR, and Scalable TCP all decrease their windows in response to a loss or mark. Modeling other congestion control algorithms might require additional assumptions to determine a lower bound on the number of in-flight packets.

\subsection{Settings Where Our Assumptions Do Not Hold}
Our results only apply when our assumptions hold. These assumptions are used by existing work on buffer sizing, so we are confident that our results are comparable to existing work.

Our results rely on the assumption that if the bottleneck queue is always non-empty, it will send exactly $\BDP$ packets every RTT. This is a simple consequence of a work-conserving queue and a constant egress line-rate. If instead we wanted to model a packet scheduling algorithm at the bottleneck link (e.g. strict precedence service, where our flows are low precedence), then we would need to determine how many packets are sent when the queue is non-empty. We could substitute that value into Assumption~\ref{assm:queue} and most of our results would apply.

There are other interesting settings for buffer sizing, especially links using AQM (e.g. \cite{FJ1993a,PNP+2013,LF2015,NJ2012a}). We believe that our techniques would help prove buffer sizing results for AQM policies, given good models of how different AQM policies interact with TCP. We believe this would be a very interesting direction for future work.

\subsection{Rules of thumb for Reno, Cubic, BBR, and Scalable TCP}
\label{sec:new-bdp-rules}

Lemma~\ref{lm:main} makes it easy to prove rules of thumb for buffer sizing. As a demonstration, here we use it to prove rules for TCP Reno and other common congestion control algorithms.

Our strategy will be to find the minimum window size and apply Lemma~\ref{lm:main} to measure link utilization. Loss based algorithms only decrease their windows on a loss, so we will use Assumption~\ref{assm:loss} (or a similar assumption for BBR) to relate the link utilization to the buffer size.

\noindent\textbf{Reno:}  Fix some time $t$ when Reno experiences a loss and let $s$ be the next RTT when $W(s) = W(t)/2$. By Assumption~\ref{assm:loss}, $W(s) \geq \frac{1}{2}(\BDP + B)$. So the utilization of the link is $\mu(s) \geq \min\left(\frac{1}{2}\inparen{1 + \frac{B}{\BDP}}, 1\right)$.

If we would like the link to be fully utilized (i.e. $\mu(s) = 1$), this result requires that $B \geq \BDP$. For smaller buffers, every 1\% reduction in utilization results in a 2\% reduction in required buffer size. So for 90\% utilization, we only need $B \geq 0.8 \BDP$. Provided our assumptions were to hold without a buffer, Reno would still have a link utilization of at least $C/2$.

This result was shown for full utilization by~\cite{Jac1990}, and for full utilization using these assumptions by~\cite{DHD2005}. We are unaware of prior results for less than full utilization.

\noindent\textbf{Decrease by $\beta$ on a loss:} Suppose a new congestion control algorithm sets its window to $W' = \beta W$ on a loss, and never decreased its congestion window otherwise. Let $t$ be the time just before a loss, and $s$ the time when $W(s) = \beta W(t).$ In this case, we would have $W(t) \geq \BDP + B$, and by Lemma~\ref{lm:main}, $\mu(s) \geq \min\left(\beta\inparen{1 + B/\BDP}, 1\right)$.

For full link utilization $\mu(s) = 1$, we would need $B > (\beta^{-1}-1) \BDP.$ This matches observations made by~\cite{HR2008,LF2015,MAK2019a}.

\noindent\textbf{Cubic:} On a loss, Cubic decreases its window by constant factor ``beta\_cubic,'' in Linux this constant is $717/1024 \sim 7/10$. Using the above result, Cubic requires $B \geq (10\gamma/7 - 1) \BDP$. For full link utilization $\gamma = 1$, Cubic requires a buffer of size $\frac{3}{7} \BDP$. This result was shown by \cite{LF2015}. For 90\% link utilization, Cubic requires a smaller buffer of size $0.28\BDP$. Without a buffer, Cubic gets a utilization of 70\%.

\noindent\textbf{Scalable:} On a loss, Scalable TCP decreases its window to $7/8$. It therefore requires a buffer size of $B \geq (8/7 \gamma - 1) \BDP$ For full link utilization, this requires $B \geq \frac{1}{7} \BDP$. For 90\% utilization, it requires a very small buffer of just $B \geq 0.03 \BDP$.  Without a buffer, Scalable gets a utilization of 87.5\%. We do not know of any prior publications of this result, but it is an easy corollary.

\noindent\textbf{BBR:} BBR (v1) alternates between phases where it probes for bandwidth and probes for the minimum RTT. We will focus on the required buffer size for full utilization during the bandwidth probing phase. BBR may behave differently in other settings, but this model characterizes the behavior we see in our experiments. 

While probing for bandwidth, BBR cycles through a series of pacing rates which limit the number of packets in flight. It picks a pacing rate of $R$ and sends at $\frac{5}{4}R$ for an RTT, $\frac{3}{4}R$ for an RTT, and $R$ for six RTTs. In our experiments, BBR encounters loss during the $\frac{5}{4} R$ phase. After seeing loss, it decreases its rate to $\frac{3}{4} R$ in response \cite[Section 4.3.4]{Car2017}.

Suppose BBR picks a pacing rate of $R \sim C$. When a loss occurs, the number of packets in-flight is at least $\BDP + B$ by Assumption~\ref{assm:loss}. During the next RTT, the bottleneck queue drains at a rate of $C-\frac{3}{4}C=\frac{1}{4}C$. After one RTT, the number of packets in-flight decreases by $\frac{1}{4} C \cdot RTT$ so $W(s) = \frac{3}{4}BDP + B$. By Lemma~\ref{lm:main}, $\mu(s) = \min\left(\frac{3}{4} + B/\BDP, 1\right)$. For full link utilization, we need $B \geq \frac{1}{4}\BDP.$ For 90\% link utilization, we need a buffer of size $B \geq 0.15 \BDP$.

We do not know of prior publications of this result. Recently there has been work showing that BBR has higher loss and is unfair to other algorithms in smaller buffers \cite{SJS+2018,CCG+2017a,CJS+2019,HHG+2018,WMSS2019a}. Our results instead focus only on the buffer size needed to keep the link fully utilized.

\noindent\textbf{BBRv2:} Google is developing a second version of BBR, and recently released an alpha implementation. We want to apply our analysis to BBRv2 as soon as the algorithm is described in sufficient detail for us to model it correctly (we assume a technical report or paper will describe the alpha code).  We do, however, include its behavior in our experimental setup using the alpha code. What we know so far is that instead of keeping the queue full and responding to losses (like all other algorithms we tested), BBRv2 keeps the queue close to empty and responds quickly to increases in delay. One consequence of this is that BBRv2 keeps the link highly utilized, but not quite 100\%. 

Hence, buffer sizing for BBRv2 is very different than for a loss-based algorithm, and BBRv2 is able to get good performance across a range of buffer sizes, while keeping the link utilization a little shy of 100\%.  This motivates us, later in this paper, to study buffer sizes that achieve high link utilization, but not quite 100\%.

\subsection{Experimental validation}

\begin{figure}
  \centering
  \includegraphics[width=0.75\linewidth]{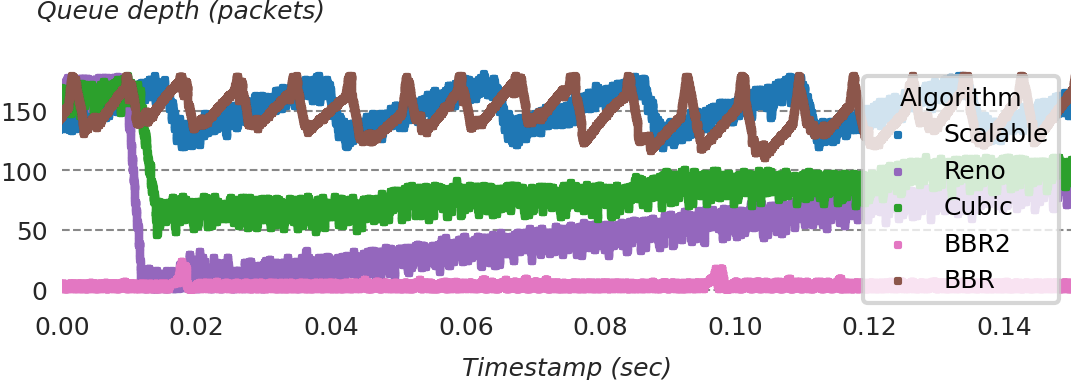}
  \caption{Queue occupancy versus time for one Reno, Cubic, BBR, and Scalable TCP flow with $B=\BDP$. The buffer is (just) large enough for Reno to keep the queue non-empty and the link fully utilized. The buffer is oversized for Cubic, BBR, and Scalable TCP which keeps the queue persistently occupied.}
  \label{fig:cubic-loss}
\end{figure}

\begin{figure*}
    \hfill
    \begin{subfigure}[h]{0.24\textwidth}
        \centering
        \includegraphics[width=\linewidth]{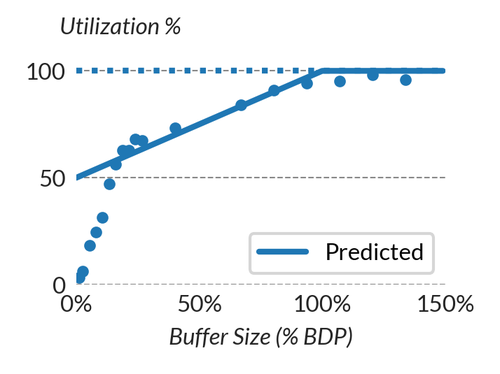}
        \caption{Reno}
    \end{subfigure}
    \hfill
    \begin{subfigure}[h]{0.24\textwidth}
        \centering
        \includegraphics[width=\linewidth]{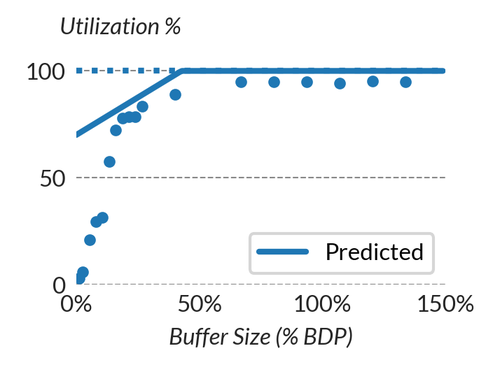}
        \caption{Cubic}
    \end{subfigure}
    \hfill
    \begin{subfigure}[h]{0.24\textwidth}
        \centering
        \includegraphics[width=\linewidth]{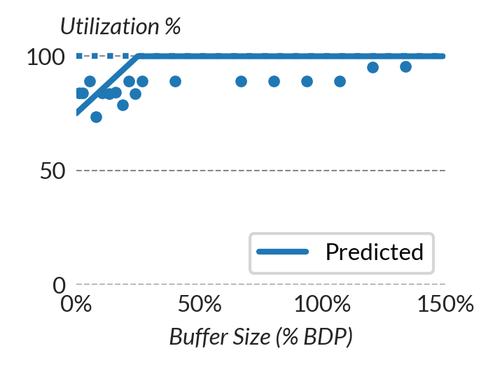}
        \caption{BBR}
    \end{subfigure}
    \hfill
    \begin{subfigure}[h]{0.24\textwidth}
        \centering
        \includegraphics[width=\linewidth]{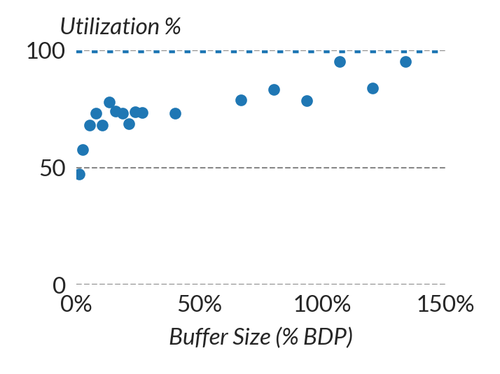}
        \caption{BBRv2}
    \end{subfigure}
    \caption{Link utilization for one TCP flow across a range of congestion control algorithms and buffer sizes.}
    \label{fig:single-flow-utilization}
\end{figure*}

Measurements from our physical network confirm the buffer sizing rules for TCP Reno, Cubic, BBR, and Scalable TCP in Section~\ref{sec:new-bdp-rules}. For example, Figure~\ref{fig:cubic-loss} shows the queue occupancy, at the moment the window is decreased. As expected, if we set the buffer size to a BDP of 165 packets, TCP Reno allows the queue to drain by a full BDP and (almost) go empty. TCP Cubic allows the queue to drain by about 60-70\%. BBR allows the queue to drain by about 25\%. Scalable TCP drains the queue to within the range 11-25\% of the BDP. The ranges we observe are consistent with the theory in Section~\ref{sec:new-bdp-rules}.

The queue depth variability in our measurements is because the sending kernel transmits bursts. Except for BBR, which uses its own pacing algorithm to reduce burstiness. 

We measure link utilization in our experiments by recording {\tt iperf's} aggregate throughput every 10ms, and report the 1st percentile value. Figure~\ref{fig:single-flow-utilization} shows the link utilization as a function of buffer size. Utilization is well-predicted by our models above a queue depth of about twenty packets. Below twenty packets utilization falls off  quickly, which appears to be caused by burstiness (bursts cause the queue size to vary by about twenty packets); i.e., burstiness can cause loss with a smaller window than Assumption~\ref{assm:loss}. Note that BBR, which has very little burstiness, has no utilization ``cliff'' below 20 packets.

\section{Sizing buffers for multiple flows}
\label{sec:multiple-flows}

In~\cite{Appenzeller:2004fk}, Appenzeller \etal argue that when $n$ TCP Reno flows share a connection, and $n$ is large, a much smaller buffer of $BDP/\sqrt{n}$ is sufficient to keep the bottleneck link highly utilized.  In particular, they prove the following theorem.
\begin{theorem}[Square root of $n$ rule]
  If for all times $t$, the windows of $n$ TCP Reno flows are independent uniform random variables in the range $\inparen{1 \pm \frac{1}{3}} \frac{\BDP + B}{n}$, and
  if $B \geq \BDP / \sqrt{n},$
  then in the limit as $n \to \infty$,
  $\P\inparen{\mu(t) < 1} \leq 0.02.$
  \label{thm:sqrt-n-rule}
\end{theorem}

Appenzeller \etal show experimental evidence that the square root of $n$ rule holds for TCP Reno, and the result has held up in later  experiments~\cite{BGG+2008}. However, Theorem~\ref{thm:sqrt-n-rule} makes the strong assumptions that: (1) flows are TCP Reno,  and (2) windows are independent uniform random variables in the range $[\frac{2}{3} \frac{BDP+B}{n}, \frac{4}{3} \frac{BDP+B}{n}]$. The rule does not apply to modern congestion control algorithms, and because of the strong assumptions it is not obvious whether a new or modified algorithm will satisfy it.

  \begin{figure*}
    \centering
    \begin{subfigure}[h]{0.32\textwidth}
        \centering
        \includegraphics[width=\linewidth]{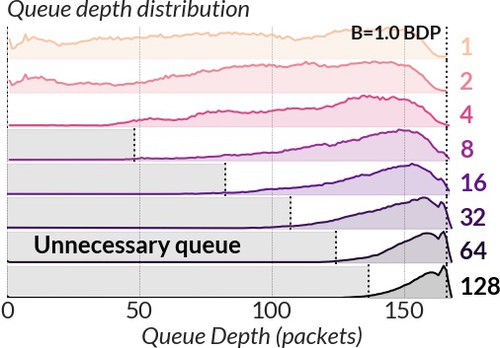}
        \caption{Buffer size: $1 \BDP$.}
    \end{subfigure}
    \hfill
    \begin{subfigure}[h]{0.32\textwidth}
        \centering
        \includegraphics[width=\linewidth]{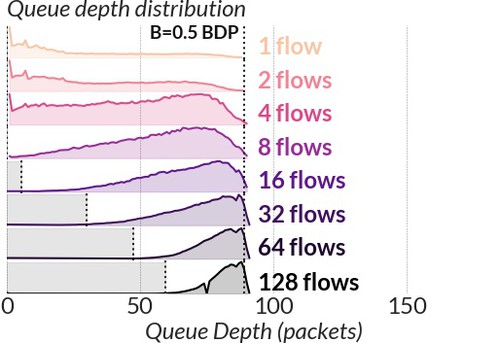}
        \caption{Buffer size: $0.5\BDP$.}
    \end{subfigure}
    \hfill
    \begin{subfigure}[h]{0.32\textwidth}
        \centering
        \includegraphics[width=\linewidth]{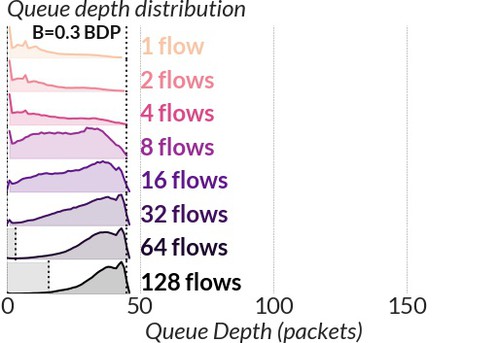}
        \caption{Buffer size: $0.25\BDP$.}
    \end{subfigure}
    \caption{Queue depth distributions for TCP Reno with decreasing buffer sizes. As the number of flows increases, the queue stays close to its maximum value (the right-most dotted line). The square root of $n$ rule predicts that queues will not fall below a certain threshold, which is shaded in grey. The non-shaded region between them shows the required buffer size for our experiments: they are a good estimate, and hold as buffer sizes decrease.
    }
    \label{fig:reno-queue-depth-distribution}
\end{figure*}

In this section, we show that the square root of $n$ result holds more broadly than for TCP Reno. 
We start, in Section~\ref{sec:relationship}, by showing that with much weaker assumptions, the aggregate window stays high when $n$ TCP Reno flows share a link. This is depicted visually in Figure~\ref{fig:reno-queue-depth-distribution}, showing how the queue depth distribution ``concentrates'' with increasing $n$. The square root of $n$ result for TCP Reno (and other multiplicative-decrease algorithms) is an immediate consequence of this, which we show in Section~\ref{sec:buffer-size}. Basically, if the aggregate window stays above $\BDP + B - \BDP/\sqrt{n}$, and $B > BDP/\sqrt{n}$, then the queue never goes empty. More broadly, we will use this effect to analyze other congestion control algorithms.

We also prove results in Section~\ref{sec:utilization} for the required buffer size when link utilization is {\em below} 100\%. In particular, we show that utilization is $\Omega(1-1/\sqrt{n})$ even with a constant size buffer. As a perhaps surprising example, if more than 1,000 flows share a link, utilization is at least 97\% even with a constant size buffer. For the same link utilization, the square root of $n$ result of \cite{BGG+2008} requires a buffer of size $\BDP/\sqrt{n}$.

Finally, in Section~\ref{sec:bbr} we extend our results to other congestion control algorithms. We show that because BBR incorporates some randomness by design, we are able to prove a square root of $n$ result with \emph{only} an assumption about fairness. Using these ideas, we propose a modification for multiplicative decrease algorithms that provably guarantees that a small buffer is sufficient, provided fairness holds.

\subsection{How Reno keeps links and queues full}
\label{sec:relationship}

TCP Reno tends to keep links and buffers full, especially as more flows use a network. Figure~\ref{fig:reno-queue-depth-distribution} shows the distribution of queue depths in our experiments. As the number of flows increases, TCP Reno keeps the queue depth distribution high. In this section, we will prove this.

In order to prove some sort of square root of $n$ rule, we will need some concept of how the aggregate window is split across multiple flows. If windows are extremely imbalanced, we cannot hope for a square root of $n$ rule. For instance, if the aggregate window is made up of only one flow, we are in the setting of the single flow case in Section~\ref{sec:single-flow} and cannot expect a square root of $n$ rule. Furthermore, in our experiments in Section~\ref{sec:eval-sqrt-n}, we find that the required buffer size gradually increases as flows become less fair.

To deal with this dependence on fairness in our theoretical results, we will define a concept of \emph{almost} fairness for TCP windows. Almost fairness has a parameter $\fairSymbol$, which relates to how far the window sizes are from equal. We will prove our results using this parameter, which gives us a convenient relationship between unfairness and buffer size.

\begin{definition}[\fair]
Consider $n$ TCP flows sharing a link. Let $w_i(t)$ be the number of in-flight packets for flow $i$ at time $t$. Let $\fairSymbol \geq 1$. Flow $i$ is {\fair} if for all time $t$,
  \begin{equation}
      \fairLower \frac{\BDP}{n} \leq w_i(t) \leq \fairUpper \frac{\BDP}{n}. \label{eq:fair-def}
  \end{equation}
\end{definition}

\fairness{} is a property of a congestion control algorithm. The closer an algorithm keeps its windows to equal, the smaller value of $\fairSymbol$ it will have. This may or may not be \emph{desirable} behavior for a congestion control algorithm, but we will be able to prove smaller buffer requirements for smaller values of $\fairSymbol$.  We can easily measure and calibrate this parameter in our experiments, and we report the results of this in Section~\ref{sec:eval-params}.
We are unaware of prior uses of this fairness metric. We could use other ways of measuring fairness, for instance Jain's fairness index \cite{JCH1998}, but \fairness{} is convenient for our proofs and easy to measure in our experiments.

It may feel somewhat unnatural that we have not defined \fairness{} as $w_i(t) \leq \fairUpper(\BDP + B)/n$, since the aggregate window is at most $\BDP + B$ and is divided among $n$ flows. However, we are thinking about the regime where $B = O(\BDP),$ and can account for the dependence on $B$ by increasing the value of $\fairUpper$ in our definition. Doing so makes our results significantly easier to prove and interpret, and the cost of a slightly looser bound on buffer size.

With this definition, we can prove our main result which relates the aggregate window after a loss to buffer size, fairness, and the number of flows. We will spend the rest of the section discussing its interpretation.

\begin{theorem}
  \label{thm:our-sqrt-n}
  Consider $n$ {\fair} TCP Reno flows sharing a link with window size $w_i(t) \geq 2$. Fix two times $0 \leq t_1 \leq t_2$. If at most $\frac{n^2}{\fairUpper \BDP} + \sqrt{n}$ flows decrease their windows between $t_1$ and $t_2$, then 
  for all $t_1 \leq t \leq t_2$
  \begin{equation*}
      W(t) \geq \BDP  + B - \frac{\fairUpper\BDP}{\sqrt{n}}. %
  \end{equation*}
\end{theorem}
\begin{proof}
See Appendix~\ref{app:proof-sqrt-n}.
\end{proof}

We have stated Theorem~\ref{thm:our-sqrt-n} for TCP Reno to make the theorem statements more direct, but as in Section~\ref{sec:single-flow}, it is easy to adapt to other congestion control algorithms. Suppose a congestion control algorithm decreases by $\beta w_i$ on a loss and otherwise increases by at least $\alpha$. We would need to add a condition to Theorem~\ref{thm:our-sqrt-n} that the smallest window $w_i \geq \alpha/(1-\beta)$. This affects Equation~\eqref{eq:aimd-bound} in the proof, but the overall result does not change.

\subsection{Buffer size required for full link utilization}
\label{sec:buffer-size}
With Theorem~\ref{thm:our-sqrt-n}, we can easily find the minimum buffer size required for full link utilization.

\begin{theorem}
\label{thm:our-sqrt-n-for-buffers}
  Consider $n$ {\fair} TCP Reno flows sharing a link with window size $w_i(t) \geq 2$. Fix two times $0 \leq t_1 \leq t_2$. If at most $\frac{n^2}{\fairUpper \BDP} + \sqrt{n}$ flows decrease their windows between $t_1$ and $t_2$, and
\begin{equation*}
    B \geq \frac{\fairUpper \BDP}{\sqrt{n}}
\end{equation*}
then $\mu(t) \geq 1$ for all  $t_1 \leq t \leq t_2$.
\end{theorem}

The key condition in Theorem~\ref{thm:our-sqrt-n} is that at most $\frac{n^2}{\fairUpper \BDP} + \sqrt{n}$ decrease their in-flight packets. This expression is unusual, and deserves some explanation.

The first term $\frac{n^2}{\fairUpper \BDP}$ is related to the number of fair flows which must decrease windows for the aggregate window to decrease. Suppose only one TCP Reno flow with a fair window size of $\BDP/n$ sees a loss and halves its window. If $\BDP/n < n-1$, then the remaining flows will each increase their windows by one and the aggregate window will \emph{increase} despite the loss. Appendix~\ref{sec:decrease-window-condition} formalizes this intuition.

The second term $\sqrt{n}$ is the number of \emph{additional} flows which decrease their in-flight packets over and above this minimum, and measures the amount of synchronization. In Section~\ref{sec:eval-sqrt-n}, we show that this is the number of additional flows which decrease their windows in our experiments. Our results also gracefully degrade as the amount of synchronization increases. We will discuss these and more modifications to Theorem~\ref{thm:our-sqrt-n} in Section~\ref{sec:desync}

Theorem~\ref{thm:our-sqrt-n-for-buffers} helps us determine the correct value of $n$ to use when sizing a buffer. Prior work has noted that Theorem~\ref{thm:sqrt-n-rule} assumes all flows are active and contributing to the buffer size, whereas in practice they might not~\cite{MAK2019b, SM2019, VSLS2007}, making it difficult to size the buffer correctly. Consider, for example, an application that starts and stops every second; should its flows be counted, or only if it recently had packets in the buffer? Theorem~\ref{thm:our-sqrt-n-for-buffers} uses the number of flows which see packet loss during an RTT, allowing us to size the buffer accordingly.

\subsection{Link utilization when $n$ TCP Reno flows share a link} 
\label{sec:utilization}
Prior buffer sizing work has looked for the smallest buffer required for full link utilization. Using Theorem~\ref{thm:our-sqrt-n}, we can also understand what happens if the buffer becomes even smaller.

\begin{theorem}\label{thm:utilization}
Consider $n$ {\fair} TCP Reno flows sharing a link with window size $w_i(t) \geq 2$. Fix two times $0 \leq t_1 \leq t_2$. If at most $\frac{n^2}{\fairUpper \BDP} + \sqrt{n}$ flows decrease their windows between $t_1$ and $t_2$, then for all $t_1 \leq t \leq t_2$ and $B \geq 0$
\begin{equation*}
    \mu(t) \geq 1 - \frac{\fairUpper}{\sqrt{n}}.
\end{equation*}
\end{theorem}

This suggests that even with very small buffers, as $n$ grows TCP Reno approaches full link utilization quite quickly, at a rate of $O(1/\sqrt{n})$. For instance if $n = 10,000$ and $\fairUpper = 2$, then the link will always be at least 98\% utilized---independent of buffer size.

This is a worst-case bound on instantaneous utilization, not on average utilization. This bounds the minimum utilization immediately after a loss. As TCP increases its window after a loss, utilization will increase and the longer term utilization will be higher.

Intuitively, the square root of $n$ result says that TCP Reno will not decrease the aggregate window too far on a loss. If the buffer is fairly large, this means the buffer will tend to remain very full. If the buffer is very small, this means that link utilization will remain high.

\subsection{BBR guarantees a $\sqrt{n}$ rule}
\label{sec:bbr}

Our square root of $n$ results so far require assumptions about how TCP flows behave. For example, Theorem~\ref{thm:our-sqrt-n}, assumes only a limited number of flows decrease their windows in each RTT.  

We can also prove a square root of $n$ rule for BBR (v1), based on how it probes for available bandwidth. If BBR detects a loss, while probing for bandwidth, it {\em randomly} decides whether or not a flow should decrease its window, which in turn desynchronizes the flows sufficiently for the  square root of $n$ rule to hold. Our proof requires no assumption on the number of flows seeing a loss at the same time. Essentially by incorporating some randomness, BBR is able to {\em guarantee} a square root of $n$ result with minimal assumptions. 

\begin{theorem}
  \label{thm:sqrt-n-bbr}
  Consider $n$ {\fair} BBR flows in the ``probe bandwidth'' phase sharing a link. Fix two times $0 \leq t_1 \leq t_2$. If
  \begin{equation}
  B > \frac{\fairUpper \BDP}{\sqrt{2 n}} \sqrt{\ln 1/\delta}, \label{eq:bbr-b-condition}
  \end{equation}
  then $\P(\mu(s) < 1) \leq \delta$ for all  $t_1 \leq t \leq t_2$.
\end{theorem}
\begin{proof}
See Appendix~\ref{sec:bbr-proof}.
\end{proof}

We believe Theorem~\ref{thm:sqrt-n-bbr} could be much stronger. Figure~\ref{fig:drop-bbr-distribution} shows that BBR's queue depths are much more tightly concentrated than for other algorithms, and more tightly concentrated than Theorem~\ref{thm:sqrt-n-bbr}. While the number of flows experiencing a loss may be about the same as with TCP Reno, fewer of them decrease their windows in BBR. We see no reason why a buffer size of $O(\BDP/n)$ would not be more appropriate for BBR, but have not been able to prove or disprove it. %

Finally, there are still many open questions about BBR's buffer behavior for future work. We have only considered BBRv1. We show BBRv2's experimental behavior in Figure~\ref{fig:drop-bbr2-distribution}, but have not analyzed it. We do not yet have a rule for BBR in the case when it decreases its window size when there is no packet loss. The main obstacle is that we do not observe this behavior in our experimental setup.

\subsection{Guaranteeing a $\sqrt{n}$ rule for a modified TCP Reno}
\label{sec:modified-reno}
Inspired by BBR, imagine that we added randomness to TCP Reno to reduce its buffer requirements (and increases its link utilization with small buffers). We are not arguing that TCP Reno should necessarily be changed; we are conducting a thought experiment to see how much it would reduce buffer requirements. 

When our imagined algorithm detects a loss, it halves its window size $w(t)$ with probability $p = 1/w(t)$. This could be done either by randomly ignoring lost packets, or by randomly marking packets with a suitable ECN policy. Let's assume that this does not interfere with {\fairness}, so that $\fairLower \BDP/n \leq w(t) \leq \fairUpper\BDP/n$, and
$p \leq 1/(\fairLower \BDP/n)$. The number of flows that randomly decide to decrease their windows in the same RTT is a sum of independent random variables, and we can use the following simple Chernoff bound to estimate the sum.

\begin{lemma}
\label{lm:num-flows-with-loss}
Let $D_i(t)$ be independent and identically distributed Bernoulli random variables, with $\E D_i = p \leq \frac{1}{\fairLower \BDP/n}$. Let $D(t) = \sum_{i=1}^n D_i(t)$ Then
$$\P\inparen{D(t) > \frac{n^2}{\fairLower \BDP} + \sqrt{n}} \leq \exp(-1/2).$$
\end{lemma}
\begin{proof}
See Appendix~\ref{sec:rand-loss-proof}
\end{proof}
This bound essentially meets the conditions for Theorem~\ref{thm:our-sqrt-n}, with the small caveat that we would need to adapt the proof to using $\fairLower$ instead of $\fairUpper$. Doing so would give us a square root of $n$ result without needing to make assumptions about which flows see a loss.

\subsection{Desynchronized flows reduce buffer requirements}
\label{sec:desync}

Before Theorem~\ref{thm:our-sqrt-n}, we knew the required buffer size for two extreme cases. When all the flows are \emph{perfectly synchronized}, they lose packets simultaneously and we need a $\BDP$ of buffering.  At the other extreme, when flows are \emph{not synchronized}, Theorem~\ref{thm:sqrt-n-rule} tells us that the buffer size can be decreased by a factor of $\sqrt{n}$. 

We can extend Theorem~\ref{thm:our-sqrt-n} to the intermediate case when $\frac{n^2}{\fairUpper \BDP} + s$ flows see a loss. In this case
\begin{equation}
    W(t) \geq \BDP + B - s \frac{\fairUpper \BDP}{n}. \label{eq:arb-additional-flows}
\end{equation}

In other words, Equation~\ref{eq:arb-additional-flows} tells us:
\begin{itemize}
    \item If flows are slightly but not {\em perfectly} synchronized, then the buffer size can be made smaller than $\BDP$. For some small constant $c$, if only $s = c n$ extra flows lose packets, then the required buffer size is $c \BDP$ which is below a $\BDP$.
    \item If only a constant number of additional flows $s$ see each loss, then a buffer of size $O(\BDP/n)$ is possible. This is tantalizing, because, of course $\BDP/n < \BDP/\sqrt{n}$. For example, if 10,000 flows share a link, and $s=1$, then we can further reduce the square root of $n$ buffer requirement by another 100-fold. To be clear, our experiments do not exhibit this degree of desynchronization, but we see it as an exciting opportunity for a new congestion control algorithm.
\end{itemize}

\section{Experiments with the $\sqrt{n}$ rule}
\label{sec:eval-sqrt-n}

We evaluated Theorem~\ref{thm:our-sqrt-n} in our physical network, using our per-packet measurement infrastructure described in Section~\ref{sec:methodology}. Our experiments show that Theorem~\ref{thm:our-sqrt-n} holds in our network, and build intuition on which factors impact buffer sizes.

\subsection{The $\sqrt{n}$ rule holds when algorithms respond to full queues.}
\begin{figure}
    \centering
    \begin{subfigure}[h]{0.49\textwidth}
        \centering
        \includegraphics[width=\linewidth]{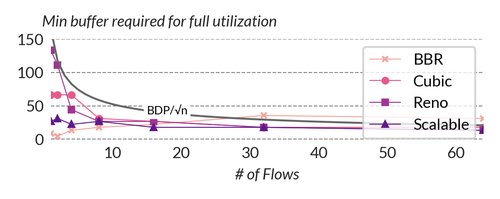}
        \caption{Buffer size required for full link utilization across all our experiments.}
        \label{fig:sqrt-n-tested}
    \end{subfigure}
    \hfill
    \begin{subfigure}[h]{0.49\textwidth}
        \centering
        \includegraphics[width=\linewidth]{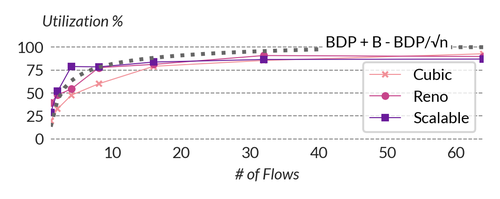}
        \caption{Utilization as a function of number of flows for a 20 packet buffer.}
        \label{fig:utilization-sqrt-n}
    \end{subfigure}
    \caption{Buffer size and link utilizations behave according to a square root of $n$ rule in our experiments}
\end{figure}

The square root of $n$ rule holds in our experiments for multiple congestion control algorithms (TCP Reno, Cubic, Scalable TCP, and BBR), despite the introduction of PRR. Figure~\ref{fig:sqrt-n-tested} shows the minimum queue depth required for full link utilization in our experiments. Each experiment lasted one thousand RTTs. We measured minimum utilization by recording {\tt iperf's} aggregate throughput over ten consecutive RTTs. We report the 1st percentile throughput.

In each experiment we found the smallest buffer size for which 99\% of packets (or more) experienced a non-empty queue. The buffer requirement follows a square-root of $n$ rule for TCP Reno. Other congestion control algorithms have similar buffer requirements for a large numbers of flows (and smaller buffer requirements for smaller numbers of flows). Theorem~\ref{thm:our-sqrt-n-for-buffers} requires that $B \geq \fairUpper \BDP/\sqrt{n}$, and our experimental results suggest it is an overestimate.

Theorem~\ref{thm:utilization} predicts that utilization will grow at a rate of $1/\sqrt{n}$ when the buffer is too small for full link utilization. Figure~\ref{fig:utilization-sqrt-n} shows the utilization as a function of the number of flows for a 20 packet buffer. Utilization increases at a rate of $1/\sqrt{n}$. It is slightly underpredicted by Theorem~\ref{thm:utilization}, and again we believe that the constants may be slightly looser than necessary.

\subsection{Algorithms keep queues full, as predicted by the $\sqrt{n}$ rule.}
\label{sec:eval-queue-distributions}

The square root of $n$ rule is primarily used to predict the minimum required buffer to keep the link utilized, as in Theorem~\ref{thm:our-sqrt-n-for-buffers}. But Theorem~\ref{thm:sqrt-n-rule} also predicts how far the queue occupancy will deviate from full if the buffer is larger than we need.  %
Figure~\ref{fig:reno-queue-depth-distribution} shows the distribution of queue occupancy for different numbers of flows (and different buffer sizes). For example, with $B=\BDP$ in Figure~\ref{fig:reno-queue-depth-distribution}(a), the distribution of queue occupancy narrows as the number of flows increase. The shaded area to the left of the distribution is unused (and therefore unnecessary) buffer; we could use a smaller buffer and shift the distribution to the left. The remaining graphs show this as we decrease the buffer size. 

We have seen that non-Reno congestion control algorithms also follow a square root of $n$ rule in our experiments, and Figure~\ref{fig:other-queue-depth-distribution} shows the distribution of queue depths for these algorithms. As with Reno, the distributions concentrate tightly around a full buffer.

\begin{figure*}
    \hfill
    \begin{subfigure}[h]{0.32\textwidth}
        \centering
        \includegraphics[width=\linewidth]{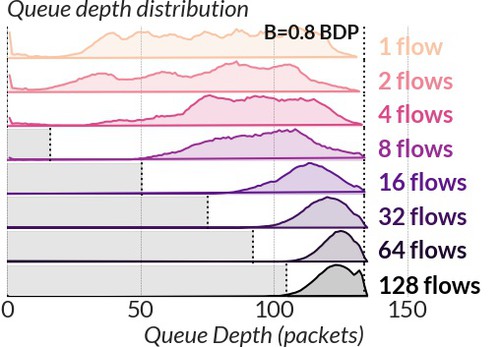}
        \caption{Cubic}
        \label{fig:drop-cubic-distribution}
    \end{subfigure}
    \hfill
    \begin{subfigure}[h]{0.32\textwidth}
        \centering
        \includegraphics[width=\linewidth]{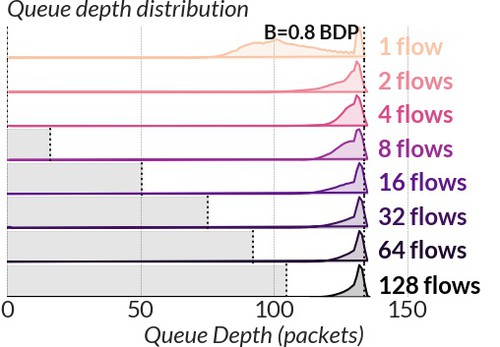}
        \caption{BBR}
        \label{fig:drop-bbr-distribution}
    \end{subfigure}
    \hfill
    \begin{subfigure}[h]{0.32\textwidth}
        \centering
        \includegraphics[width=\linewidth]{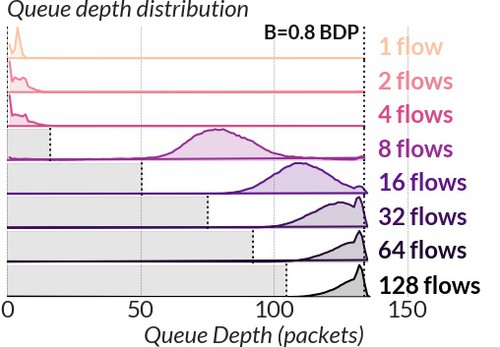}
        \caption{BBRv2}
        \label{fig:drop-bbr2-distribution}
    \end{subfigure}
    \caption{Queue depth distribution for various congestion control algorithms and buffer size $0.8 \BDP$. The non-shaded area corresponds to a buffer size of $2\BDP/\sqrt{n}$, which tends to overestimate the required buffer.}
    \label{fig:other-queue-depth-distribution}
\end{figure*}

\begin{figure}[t]
\begin{minipage}{0.49\textwidth}
  \includegraphics[width=\linewidth]{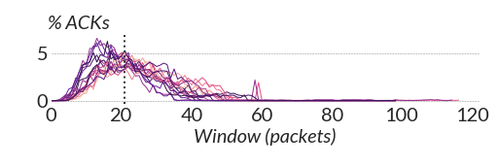}
  \caption{Distribution of window sizes measured for each received acknowledgement for 16 TCP Reno flows over one thousand RTTs with a $\BDP$ sized buffer.}
  \label{fig:fairness}
\end{minipage}%
\hfill%
  \begin{minipage}{0.49\textwidth}

    \includegraphics[width=\textwidth]{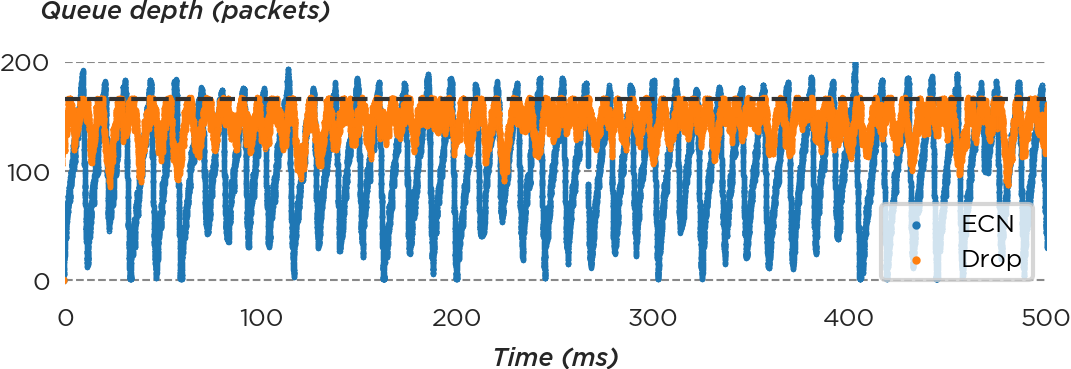}
    \caption{Queue occupancy over time for 16 TCP Reno flows sharing a queue. Packets are dropped or marked when the queue exceeds one BDP. In the ECN trace, more flows are marked when the queue is exceeded and flows become more synchronized.}
    \label{fig:ecn-reno-trace}
  
  \end{minipage}
  \end{figure}

\subsection{The $\sqrt{n}$ rule holds with non-uniform window distributions.} Figure~\ref{fig:fairness} shows the distribution of windows for 16 TCP Reno flows over one second, when $B=\BDP$. The distributions are clearly not uniform (and we observe windows in the ranges of 0-10 and 40-60 packets, outside the $\frac{3}{4}$ and $\frac{5}{4}$ range). This suggests that the uniform window size assumption in \cite{Appenzeller:2004fk} does not hold in practice.

Figure~\ref{fig:fairness} also shows the degree to which flows are treated fairly, in the sense that the the average window size for each flow is close to $(\BDP+B)/n$. It is not surprising that they are similar, since one of the goals of a congestion control algorithm is to fairly share a bottleneck link. However, we can see that the flows do not have identical window size distributions, with some average window sizes two to three times larger than the fair share of $(\BDP+B)/n$, and some half the size. This is why we have defined {\fairness}.

\subsection{The $\sqrt{n}$ rule does not hold when losses are synchronized.}
\label{sec:ecn-experiment}
Theorem~\ref{thm:our-sqrt-n} requires flows to be desynchronized, in that not too many flows can respond to loss at the same time. This is a necessary requirement and we would not expect the square root of $n$ rule to hold when flows are synchronized. In the extreme, suppose all flows experience a loss between times $t_1$ and $t_2$. In the worst case $W(t_2) = \sum_{i=1}^n \frac{1}{2} w_i(t_1) = \frac{1}{2} W(t_1)$. As pointed out in prior work~\cite{Appenzeller:2004fk,DHD2005}, we are back in the same situation as for one flow and require a BDP of buffering.

While larger numbers of TCP Reno flows tend to be less synchronized in our experiments, this is not the case for all algorithms or in all settings. As a stark illustration, we describe an experiment in which synchronization caused a larger buffer requirement. Our experiment compared dropping packets to a specific way of ECN marking: we dropped or marked packets whenever the queue depth exceeded a fixed threshold.

We ran this experiment for sixteen TCP Reno flows. Figure~\ref{fig:ecn-reno-trace} shows the queue depth over time during this experiment, comparing TCP Reno with ECN (blue line) and with dropped packets (orange line). It is clear that with ECN marking, the queue depth fluctuates much more, hence requiring a much larger buffer. With ECN the queue regularly overshot its marking threshold, so more packets were marked in the same RTT, which led to a synchronized reduction in window size.

We would like to emphasize that this experiment should \emph{not} be taken as a general statement about buffer sizing with ECN, since ECN is often used in conjunction with an AQM policy (e.g. RED~\cite{FJ1993a}). However, ECN is used in this way for BBRv2 \cite{CCY+2019b}, DCTCP \cite{AGM+2010}, and in other settings~\cite{TMR2019, RJ1990}. %

\subsection{The $\sqrt{n}$ result does not hold when flows are too unfair.} Theorem~\ref{thm:our-sqrt-n} requires windows to be almost fair, and this is a necessary requirement. Figure~\ref{fig:unfairness} illustrates how unfairness between flows can affect the required buffer size. Consider the trace of four TCP Reno flows sharing a link, particularly the large red window and the two smaller purple windows (superimposed). When the small purple windows experience losses simultaneously at 725ms, both halve their windows from 80 to 40 packets and the buffer occupancy drops by 65 packets. Yet, on its own, the larger red window causes the buffer to drop by the same amount when it halves its window size -- this is because it starts from a larger size. In general, bigger windows have a disproportionately larger effect on the buffer occupancy. In the extreme, imagine one flow used {\em all} of the link capacity, and the remaining $n-1$ flows used none. Such extreme unfairness means we are back in the single flow setting, and we can not do better than the $\BDP$ rule.

\begin{figure}
\centering
  \includegraphics[width=0.49\linewidth]{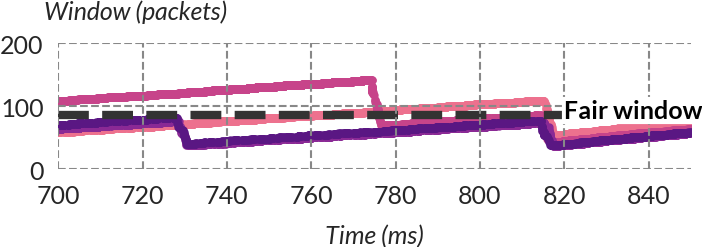}
  \includegraphics[width=0.49\linewidth]{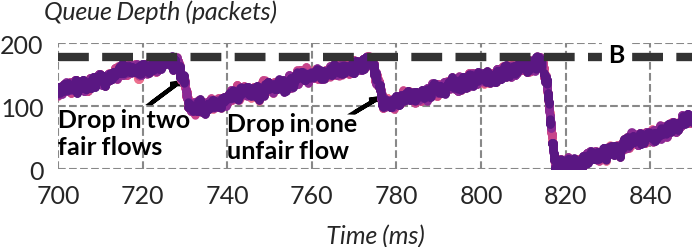}
  \caption{Windows and queue depths over time for four TCP Reno flows sharing a link. A loss for a flow with an unfairly large window results in the same queue depth decrease as for two fair flows, and half as large a decrease as a loss for all flows.}
  \label{fig:unfairness}
\end{figure}

\subsection{Checking theorem conditions in our experiments.}
\label{sec:eval-params}
Theorem~\ref{thm:our-sqrt-n} depends on the value of {\fairness}, and the number of flows which see a loss in a RTT. These are not well known values, and so we use our experimental setup to measure them and give a sense of their magnitude.

First, we will check whether the number of additional flows which see a loss is about $\sqrt{n}$. This is a strong assumption in general, but it appears to be a loose upper bound in our experiments. We ran a number of experiments with a buffer size of 1 BDP and increasing numbers of TCP Reno flows. We split each experiment into RTTs, and for each RTT where there was one loss, counted the number of flows which saw a loss in each RTT. Figure~\ref{fig:prob-loss-error} shows the average number of flows that observed a loss in that RTT. The bound $\frac{n^2}{2 \BDP} + \sqrt{n}$ of Theorem~\ref{thm:our-sqrt-n} is also plotted. It appears to be a good estimate for a small number of flows, and an \emph{overestimate} for larger numbers of flows.

Next, we check the values of {\fairness} in our experiment. To measure {\fairness}, we find the 1st and 99th percentile window sizes for all flows in an experiment. We calculate $\BDP/n$ from the experiment parameters. We then find the smallest value of $\Delta$ so that the definition of {\fair} is satisfied between the 1st and 99th percentile window sizes.

We first look at how $\fairSymbol$ depends on the number of flows. Figure~\ref{fig:delta-fairness-num-flows} shows that both $\fairLower$ and $\fairUpper$ slowly increase with the number of flows. We next look at how changing the buffer size impacts fairness. We look at all experiments with sixteen flows, to avoid confounding with the dependence on the number of flows. We find that $\fairLower$ slowly decreases with buffer size, and approaches 0.5. For Cubic, Reno, and Scalable, $\fairUpper$ is about 4. For BBR, it is about 10. $\fairUpper$ does not depend on the buffer size for any of these algorithms.

\begin{figure}
  \centering
  \includegraphics[width=0.75\textwidth]{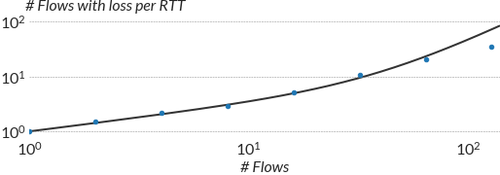}
  \caption{Fraction of TCP Reno flows which see a loss during loss events, for various numbers of simultaneous flows. Each point shows the mean number of flows seeing a loss during each RTT, and the line shows the bound of $\frac{n^2}{2\BDP} + \sqrt{n}$ from Theorem~\ref{thm:our-sqrt-n}.
  }
  \label{fig:prob-loss-error}
\end{figure}

\begin{figure}
\centering
\begin{subfigure}[t]{0.49\textwidth}
    \includegraphics[width=\linewidth]{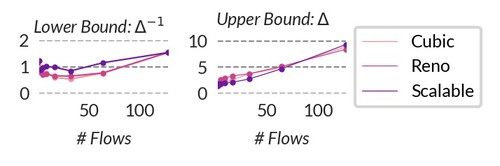}
    \caption{Varying numbers of flows}
    \label{fig:delta-fairness-num-flows}
\end{subfigure}
\hfill
\begin{subfigure}[t]{0.49\textwidth}
    \includegraphics[width=\linewidth]{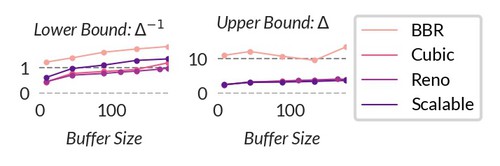}
    \caption{Varying buffer size in packets}
    \label{fig:delta-fairness-buffers}
\end{subfigure}
\caption{Worst-case {\fairness} in our experiments.}
\end{figure}
\section{Supporting evidence from real-world measurements}
\label{sec:in-the-wild}

Square root of $n$ results hold because the occupancy of router queues "concentrates" during times of congestion, leading to persistently occupied standing queues. This phenomenon has been widely observed outside of buffer sizing work. For example, Lee \etal~\cite{LLYM2011a} report that operating a link at 100\% utilization resulted in a persistent 30-50ms increase in queueing delay. In~\cite{Tah2017}, T\"aht and Reed report on an outage where they observed a 400ms increase in RTTs. In \cite{CSB+2015}, Chandrasekaran \etal report 20ms increases in RTTs over congested links in the internet core. 

Recent measurement work uses these increases in RTT to measure congestion. Luckie \etal propose TSLP~\cite{LDC+2014}, which measures RTTs periodically and uses a persistent increase in RTT as an indicator of congestion. Their figures show that congestion can elevate minimum RTTs by tens of milliseconds. In \cite{FVD2017}, Fanou \etal use TSLP to measure congestion in the African IXP substrate, and find links with elevated minimum RTTs during congestion. Sundaresan \etal classify flows with a tightly concentrated RTT during slow-start as flows experiencing congestion~\cite{SADc2017}. In \cite{DCG+2018}, Dhamdhere \etal measure the minimum RTT over a period of five and fifteen minutes to identify links experiencing persistent congestion.

In \cite{SWH+2019}, Spang \etal report on buffer experiments run at Netflix. They show that large buffers can increase RTT for \emph{all} traffic sharing a link. They run an experiment where a subset of Netflix traffic is randomly assigned to two routers, one with a 500MB buffer and one with a 50MB buffer. They observe that the minimum RTT increase by hundreds of milliseconds for all flows; and gets worse with larger buffer sizes. This suggests the queue remains relatively full for the entire hour.

Our results explain why congestion control algorithms that respond to full queues (via losses or marks) create standing queues during times of congestion. Our experiments in Section~\ref{sec:eval-queue-distributions} show the same persistent standing queues, especially as the number of flows increases. 

All of the measurements described above show evidence of concentration. While they do not prove that the concentration is of the order $\BDP / \sqrt{n}$, our results suggest that all these buffers could be shrunk without impacting link utilization. Our results also suggest that link utilization would remain high in these cases, even with a very small buffer.

It is natural to ask whether concentration relies on current congestion control algorithms, and whether a new congestion control algorithm might break it. For instance, if a few large sources of traffic switched to BBR, would congestion on the internet look very different? Our results in Section~\ref{sec:multiple-flows}, especially the ones for BBR, suggest it would not. But we cannot rule out some other, future congestion control algorithm with very small queues during congestion.

Finally, our results suggest intriguingly that there might be "dark congestion" in the internet, which current measurement techniques that look for persistent standing queues cannot detect. There might be routers that cause a large amount of synchronization, for instance by marking packets as in Section~\ref{sec:ecn-experiment}. This would lead to wildly varying queue occupancy, making it hard to detect congestion.

\section{Recommendations}
\label{sec:recommendations}

\subsection{Recommendations for Operators}
\label{sec:operators}

The networking community as a whole, both industry and academia, knows surprisingly little about buffer size requirements, and tends to oversize buffers in the WAN. We believe there might be big benefits to network operators if they run experiments to determine the required buffer size in their network. Consider, for example, the difference in required values for a network with a BDP of $100ms \times 100Gb/s = 10^{10}$ bits carrying 10,000 TCP Reno flows. If we adopt the single flow rule, we need a 10Gbit buffer; if Theorem~\ref{thm:our-sqrt-n} applies we might need only 100Mbit buffers. BBR only needs a 100Mbit buffer, and BBRv2 might require even less. If we were happy with an instantaneous utilization higher than 99\%, the buffer could be reduced a few tens of packets. This could make it possible to use simpler, cheaper routers with on-chip buffering.

We should point out that this paper focuses exclusively on the relationship between sizing a buffer and link utilization. Buffers can also have a large impact on application performance \cite{DHD2005,DD2006,HPC+2014,SWH+2019,BLG2019,ASB+2019}, and it is  important to consider these factors. More experiments would benefit operators (and the broader networking community) by improving our understanding of how buffers impact application performance.

Theorem~\ref{thm:our-sqrt-n-for-buffers} may be easy to verify in practice. If end hosts can be instrumented, one could measure the total decrease in windows over a short period to estimate the required buffer size. With fine-grained metrics from the switch, it would be possible to count the number of unique flows dropped during congestion events (e.g. by forwarding all lost packets to a collector), estimating the bandwidth of these flows (e.g. using coarse netflow statistics), and size the buffer accordingly. This could be done manually, or automated and be done periodically. If only loss statistics are available, it would be possible to estimate the number of flows which see simultaneous loss using simple probability calculations. Then Equation~\eqref{eq:arb-additional-flows} could be used to size the buffer. Note that it is \emph{not} necessary to measure the total number of flows to apply our results, which avoids the problems pointed out in \cite{MAK2019b, SM2019, VSLS2007}.

\subsection{Recommendations for Congestion Control Algorithm Designers}

To someone working on congestion control, this paper may feel a bit backwards. As argued by~\cite{MM2019a}, instead of sizing buffers based on unintended artifacts of TCP Reno (i.e. as in the single flow rule), we could instead design congestion control algorithms that work well for all queue sizes. If we want full link utilization and a small buffer, the congestion control algorithm only needs to prevent the aggregate window from dropping below a $\BDP$.

In Section~\ref{sec:single-flow}, we showed that smaller reductions in window size lead to smaller buffer requirements. In Section~\ref{sec:bbr} and Section~\ref{sec:modified-reno}, we showed how adding a small amount of randomness to a congestion control algorithm can reduce its buffer requirements. We encourage congestion control algorithm designers to use these techniques to reduce buffer requirements. Our experimental results suggest that BBRv2 is a good first step in this direction. With a bit of work, new algorithms can be more friendly to buffers, and the people operating them.

\section{Related Work}
\label{sec:related}

There have been many papers published about buffer sizing in the past twenty years. We will focus on work most closely related to our results, for broader surveys see~\cite{VST2009a,MAK2019b}.

The rule of thumb for sizing router buffers with a single TCP Reno flow is attributed to \cite{Jac1990,VS1994}. The rule was extended to multiplicative-decrease algorithms by \cite{HR2008,LF2015}. The proof of the BDP rule by \cite{DHD2005} did not depend on silence after a loss. We believe our discussion of PRR and BBR are new.

For multiple flows \cite{Appenzeller:2004fk} introduced the square root of $n$ rule, which was tested experimentally by \cite{BGG+2008} in the Level 3 backbone. \cite{DHD2005} described a rule incorporating a model of loss, and includes similar style of analysis as Theorem~\ref{thm:our-sqrt-n}. \cite{WG2006} show that number of flows which see a loss is important for whether a square root of $n$ rule applies.

TCP's tendency to observe loss simultaneously and synchronize has been observed in simulation since the 1980s~\cite{Has1989,Zha1989,ZC1990,Zhang:1991bqa,FJ1994,JRF+2001}. More recently,~\cite{LF2015} observed synchronization using physical hardware. Although we only observe a small amount of synchronization in our experiments, our buffer size results can handle larger amounts. We think it is an interesting open question why it seemed to be more common in the 80s and 90s, and why it only sometimes appears today.

In the early 2010s, \cite{GLN2011} revisited the question of buffer sizing, observing that the oversized buffers in cable modems cause massive delays and standing queues during congestion, often referred to as {\em bufferbloat}. Bufferbloat usually refers to home networks where a small numbers of flows typically share the buffer, while we are more focused on backbone routers with a large number of flows; but we view this work as complimentary. We discuss how our results are related to standing queues in Section~\ref{sec:in-the-wild}, albeit for larger numbers of flows, and our results may shed light on other parts of the Internet where ``Dark Buffers'' lurk.

Our work focuses on the relationship between buffer size and link utilization. A separate, very interesting line of work has explored other impacts of buffer sizing. \cite{DHD2005,DD2006} discuss impacts on TCP such as increased loss and throughput variability. A recent line of work \cite{SJS+2018,CCG+2017a,CJS+2019,HHG+2018,WMSS2019a} has shown that BBR can have large loss and worse fairness in small buffers. \cite{HPC+2014} describes a test-bed study showing that buffer sizing can cause a significant change in application QoE. \cite{SWH+2019} reports on production buffer sizing experiments at Netflix, and show that buffer size has a large impact on video performance. \cite{BLG2019} reports on buffer experiments at Facebook, including impacts on flow completion time. \cite{ASB+2019} shows an example where shrinking the buffer size of a home WiFi router can degrade video performance.

\section{Conclusion}
\label{sec:conclusion}

The main takeaways from this paper are new results on buffer sizing; in particular, a better understanding of the relationship between link utilization and buffer size, and how congestion control algorithms can impact this relationship. Prior work suggested that buffers can be reduced by a factor of square root $n$, offering dramatic buffer reductions in networks carrying many flows. However, the result required TCP Reno, did not make clear how we determine $n$, or what happens when windows are not independent. Our results clarify that $n$ is the number of flows that reduce their window size in the same RTT, removes the need for independence, and holds for a broader class of congestion control algorithms. Our results explore what happens when buffers are sized too small for full link utilization. Our results make it easier to run buffer sizing experiments, which should give much more confidence that the results apply broadly.

There remains work to be done. For instance, we assumed that RTTs are the same for all flows, yet clearly this is not true in practice. Variance among RTTs should only reduce synchronization, and hence further reduce the buffer requirement, but we have not been able to prove it. We still do not understand the underlying causes of synchronization, or how to reduce it in drop-tail queues. We believe our BBR result can be improved, and that BBR may allow a very small buffer. We have focused on sizing the buffer for the worst-case behavior over a relatively long period of time, and it may be possible to dynamically adjust the buffer size on a much shorter timescale.

More generally, though, we still have only a rudimentary understanding of buffer size, despite its potentially big impact on application performance. Even if we know all the traffic using a link, we don't know how to predict the best buffer size. We encourage further experimentation and measurement with real applications.

Finally, we believe that future congestion control algorithms can significantly reduce buffer requirements.
With small modifications, existing algorithms can reduce buffer requirements, or increase link utilization with small buffers.
If packet arrivals can be managed perfectly, buffers can be made smaller or even eliminated all together. %
We see no reason why a future congestion control algorithm would need anything more than a very small buffer.

\section{Acknowledgements}

We would like to thank Vladimir Gurevich for all his incredible support around debugging P4. We would also like to thank Joseph Little for the invaluable IT support. This research was funded by Netflix. The opinions expressed in this article are the authors’ own and do not reflect the view of Netflix.

\bibliographystyle{IEEEtran}
\bibliography{refs,zotero-refs}

\appendix
\section{Proof of Theorem~\ref{thm:our-sqrt-n}}
\label{app:proof-sqrt-n}
\begin{proof}
Fix some time $t$ such that $t_1 \leq t \leq t_2$. Let $D$ be the set of flows decreasing their windows by time $t$, let $w = W(t_1)$ and $w' = W(t)$. After $D$ has decreased their windows, the aggregate window $w'$ is
\begin{equation*}
    w' = w - \sum_{i \in D} \frac{1}{2} w_i(t_1) + \sum_{i \in \bar{D}} 1.
\end{equation*}
By a union bound,
\begin{equation*}    
    w' \geq w - \frac{1}{2} |D| \max_{i \in D} w_i(t_1) + n - |D|.
\end{equation*}
Let $w^* = \frac{\fairUpper\BDP}{n}$. By \fairness, $\max_{i \in D} w_i(t_1) \leq w^*$. Substituting and simplifying we have,
\begin{equation}
    w' \geq w + n - |D|\left(1+w^*/2\right), \label{eq:aimd-bound}
\end{equation}
Since $w_i(t) \geq 2$, $w^* \geq 2$, and so $1+w^*/2 \leq w^*.$ Therefore,
\begin{equation*}
       w' \geq w + n - |D|w^*.
\end{equation*}
By the conditions of the theorem $|D| \leq \frac{n^2}{\fairUpper \BDP} + \sqrt{n} = \frac{n+\sqrt{n}w^*}{w^*}.$ Substituting, we have
\begin{equation*}
       w' \geq w - \sqrt{n}w^*.
\end{equation*}
Expanding the definition of $w^*$, this becomes
\begin{equation*}
       w' \geq w - \fairUpper\frac{\BDP}{\sqrt{n}}.
\end{equation*}
By Assumption~\ref{assm:loss},
\begin{equation*}
       w' > \BDP + B -\fairUpper\frac{\BDP}{\sqrt{n}}.
\end{equation*}
\end{proof}

\section{Proof of Theorem~\ref{thm:sqrt-n-bbr}}
\label{sec:bbr-proof}

The proof of a square root of $n$ rule for BBR takes a slightly different path from the usual AIMD proof. In particular, we don't need to assume that the number of flows which see a loss is limited.

During the probe bandwidth phase, BBR picks a pacing rate of $R$ and sends at pacing rates $\{\frac{5}{4} R, \frac{3}{4} R, R, R, R, R, R, R\}$ for an RTT each. It begins this cycle at a random point (which is not $\frac{3}{4} R$).

When BBR experiences a loss, it only decreases its rate if it is sending at a pacing rate of $\frac{5}{4} R$ \cite[Section 4.3.4]{Car2017}. In expectation, if flows were equally split between all the pacing rates, we would only expect about $n/8$ flows to decrease their rates by $1/4$. And even better, another $n/8$ flows would increase their rates by $1/4$ over the next RTT and the total change would be small.

\begin{proof}
In the worst case, \emph{all} BBR flows see a loss at time $t$. Consider some flow $i$, and let $R_i(t)$ be the rate of flow $i$ at time $t$, and $X_i(t) = R_i(t) \cdot \RTT - R_i(t_1) \cdot \RTT$. Because of the random order of pacing rates, the $X_i(t)$ are independent and identically distributed according to
$$X_i(t) = \begin{cases} -\frac{1}{4}R_i(t_1)\cdot\RTT  &\mbox{with probability } \frac{1}{8} \\
  +\frac{1}{4}R_i(t_1)\cdot\RTT  &\mbox{with probability } \frac{1}{8} \\
  0  &\mbox{with probability } \frac{6}{8}\end{cases}$$
  
Our goal will be to bound $\sum_{i=1}^n X_i(t)$ and apply the usual argument with Lemma~\ref{lm:main}.

Note that $\E X_i(t) = 0$. Let $c_i = \inabs{X_i(t) - \E X_i(t)} = \inabs{X_i(t)} = R_i(t_1)\cdot\RTT/4$. By {\fairness}, $c_i$ is at most $\frac{\fairUpper \BDP}{4n}.$ Applying Azuma-Hoeffding,
\begin{align*}
    \Pr\inparen{\sum_{i=1}^n X_i(t) < t} &\leq \exp\inparen{-\frac{t^2}{2 \sum_{i=1}^n c_i^2}},\\
    &\leq \exp\inparen{-\frac{t^2 2n}{\fairUpper^2 \BDP^2}}.
\end{align*}
Fix some $\delta > 0$. If $t = \frac{\fairUpper \BDP \sqrt{\ln 1/\delta}}{\sqrt{2n}}$, then
$$\Pr\inparen{\sum_{i=1}^n X_i(t) < t} \leq \delta.$$

By definition of $X(t)$, $W_i(t) = \BDP + B - X(t)$. We have just shown that with probability at least $1-\delta$, $W_i(t) \geq \BDP + B - t$. Therefore by Lemma~\ref{lm:main}, if $B > t$ then $Q(t) > 0$ with probability at least $1-\delta$.
\end{proof}

\section{Minimum number of flows which must decrease windows}
\label{sec:decrease-window-condition}

Many of our results have a condition on the number of flows which decrease their windows that includes an expression like $\frac{n^2}{\fairUpper \BDP}.$ This is related to the number of flows which must decrease their windows in order for $W(t+1)$ to be lower than $W(t)$.

To see this, let $D(t)$ be the set of Reno flows which decrease their windows during RTT $t$. Then the window in the next RTT is
$$W(t+1) = W(t) - \frac{1}{2}\sum_{i \in D(t)} w_i(t) + n - |D(t)|.$$
Let $w^-$ be the smallest window in $D(t)$. By a union bound, we have
$$W(t+1) \leq W(t) - \frac{1}{2}|D(t)|w^- + n - |D(t)|.$$
In order for $W(t+1) \leq W(t)$, we need
$$|D(t)| \geq \frac{n}{1+w^-/2}.$$
If we let $w^+ = \frac{\fairUpper\BDP}{n}$, the condition we have on $|D(t)|$ in the theorem statements is
$$|D(t)| \geq \frac{n}{w^+}.$$
There are two differences between the two. Since we want a lower bound on $W(t+1)$ to bound the minimum buffer size, the condition has $w^+$ instead of $w^-$. We also chose to use a slightly looser bound of $(1+w^+/2) \leq w^+$ in \eqref{eq:aimd-bound} to simplify the required buffer size.

\section{Proof of Lemma~\ref{lm:num-flows-with-loss}}
\label{sec:rand-loss-proof}
The proof of Lemma~\ref{lm:num-flows-with-loss} is a standard application of a Chernoff bound.
\noindent \begin{proof}
First, we calculate $\E D(t)$,
\begin{equation}
    \E D(t) = np \leq \frac{n^2}{\fairUpper\BDP} \label{eq:exp-d-t}
\end{equation}
Note that $D(t) - \E D(t)$ is the sum of independent random variables, and each term is at most $1$ in absolute value since
$$|D_i(t) - \E D_i(t)| \leq \max(1-p, p) \leq 1.$$
Applying Azuma-Hoeffding, for $\eps > 0$ we have
\begin{align*}
    \P\inparen{D(t) - \E D(t) > \eps} &\leq \exp\inparen{-\frac{\eps^2}{2n}}.
\end{align*}
Plugging in \eqref{eq:exp-d-t} and $\eps = \sqrt{n}$, we have the desired result:
\begin{align*}
    \P\inparen{D(t) > \frac{n^2}{\BDP} + \sqrt{n}} &\leq \exp\inparen{-\frac{n}{2n}},\\
    &\leq \exp(-1/2).
\end{align*}
\end{proof}
We could get a higher probability bound, for instance for any $\delta > 0$, by increasing the upper bound on $D(t)$ to $\frac{n^2}{\fairUpper \BDP} + \sqrt{2n\ln(1/\delta)}$. This would change the assumption in Theorem~\ref{thm:our-sqrt-n}, and one could adapt the proof of Theorem~\ref{thm:our-sqrt-n} accordingly.

\end{document}